\documentclass[letterpaper]{article}
\usepackage{amsmath}
\usepackage{aaai2026}
\usepackage{amsthm}

\newtheorem{theorem}{Theorem}

\usepackage{times}
\usepackage{helvet}
\usepackage{courier}
\usepackage[hyphens]{url}
\usepackage{caption}
\usepackage{natbib}
\usepackage{tikz}
\usepackage{pdfpages}
\usetikzlibrary{positioning, arrows.meta, shapes.multipart, fit, calc}
\usepackage{float}
\usepackage{graphicx}
\usepackage{multirow}
\usepackage{adjustbox}
\usepackage{subcaption}
\usepackage{textcomp}

\usepackage[ruled,vlined]{algorithm2e}
\usepackage{tcolorbox}
\usepackage{algpseudocode}
\usepackage{amsthm}
\usepackage{amsfonts}
\usepackage{booktabs}
\usepackage{xcolor}
\usepackage{amsmath}
\usepackage{float}
\begin{document}

%
%
%
\title{PRISM: Privacy-Aware Routing for Adaptive Cloud–Edge LLM Inference via Semantic Sketch Collaboration}

\title{PRISM: Privacy-Aware Routing for Adaptive Cloud--Edge LLM Inference via Semantic Sketch Collaboration}

\author{
Junfei Zhan\textsuperscript{\rm 1}\equalcontrib,
Haoxun Shen\textsuperscript{\rm 1}\equalcontrib,
Zheng Lin\textsuperscript{\rm 2},
Tengjiao He\textsuperscript{\rm 3}\thanks{Corresponding author.}
}

\affiliations{
\textsuperscript{\rm 1}Department of Electrical and Systems Engineering, University of Pennsylvania, Philadelphia, PA, USA\\
\textsuperscript{\rm 2}Department of Electrical and Electronic Engineering, University of Hong Kong, Hong Kong SAR, China\\
\textsuperscript{\rm 3}College of Information Science and Technology, Jinan University, Guangzhou, China\\
\{zjf2024, haoxun\}@seas.upenn.edu, linzheng@eee.hku.hk, htj2018@jnu.edu.cn
}

\maketitle

\begin{abstract}

Large Language Models (LLMs) demonstrate impressive capabilities in natural language understanding and generation, 
but incur high communication overhead and privacy risks in cloud deployments, while facing compute and memory constraints when confined to edge devices.
Cloud–edge inference has emerged as a promising paradigm for improving privacy in LLM services by retaining sensitive computations on local devices.
However, existing cloud–edge inference approaches apply uniform privacy protection without considering input sensitivity, resulting in unnecessary perturbation and degraded utility even for non-sensitive tokens. 
To address this limitation, we propose Privacy-aware Routing for Inference with Semantic Modulation (PRISM), a context-aware framework that dynamically balances privacy and inference quality. 
PRISM executes in four stages: (1) the edge device profiles entity-level sensitivity; (2) a soft gating module, also on the edge, selects an execution mode -cloud, edge, or collaboration; (3) for collaborative paths, the edge applies adaptive two-layer local differential privacy based on entity risks; and (4) the cloud LLM generates a semantic sketch from the perturbed prompt, which is then refined by the edge-side small language model (SLM) using local context.
Our results show that PRISM consistently achieves superior privacy-utility trade-offs in various scenarios, reducing energy consumption and latency to 40–50\% of baseline methods such as Uniform and Selective LDP, while maintaining high output quality under strong privacy constraints. 
These findings are validated through comprehensive evaluations involving realistic prompts, actual energy measurements, and heterogeneous cloud–edge model deployments.

\end{abstract}
%


\begin{small}
\begin{links}
  \link{Code}{https://github.com/Junfei-Z/PRISM}
\end{links}
\end{small}

%




\begin{figure}[htbp]
    \centering
    \includegraphics[width=0.90\linewidth]{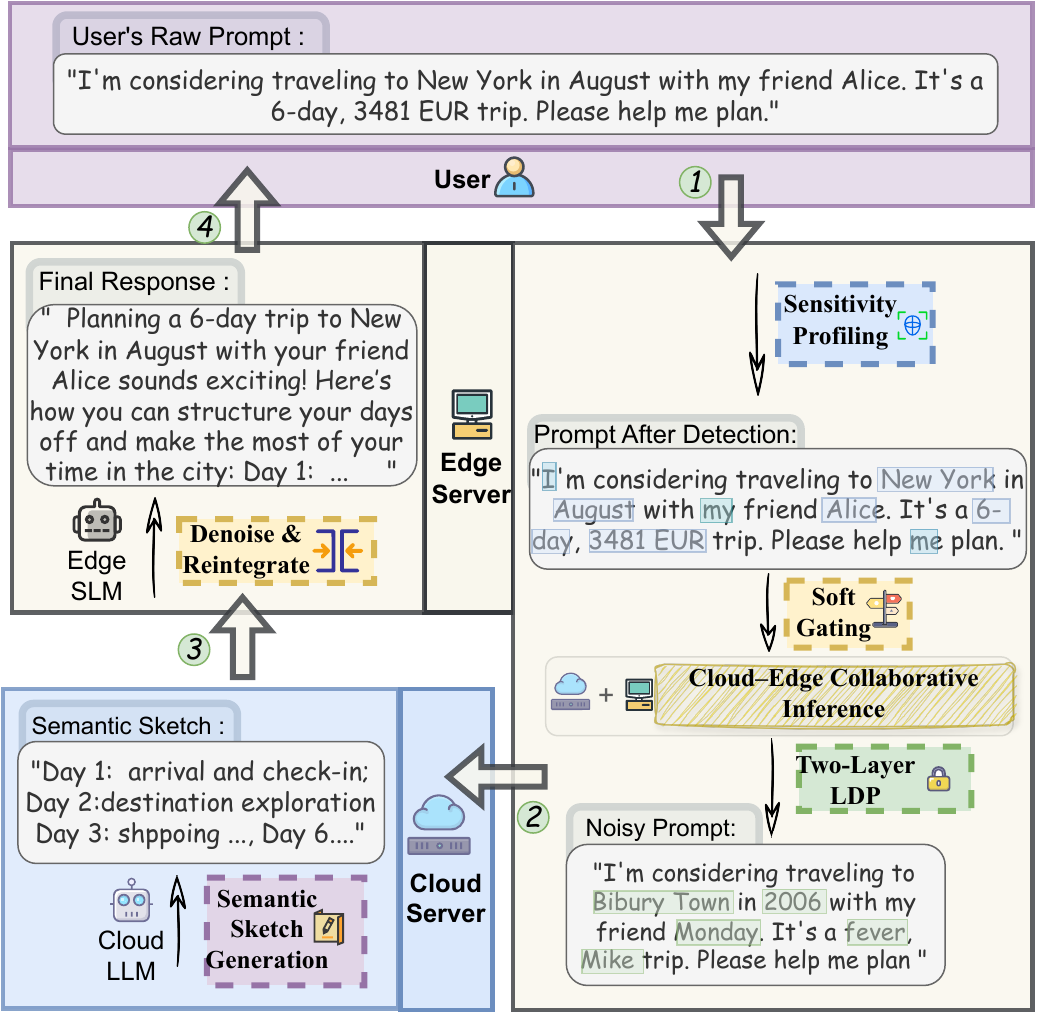}
    \caption{PRISM workflow example illustrating privacy-preserving prompt processing and transformation stages.}
    \label{fig:enter-label}
\end{figure}

\section{Introduction}
Large Language Models (LLMs) such as GPT‑4o~\cite{openai2024gpt4o}, LLaMA 3~\cite{grattafiori2024llama3}, and Qwen 3~\cite{yang2025qwen3} have achieved remarkable progress in natural language understanding and generation~\cite{cai2025graph}
. Their capabilities extend to diverse applications including code synthesis~\cite{chen2021evaluating}, multimodal reasoning~\cite{alayrac2022flamingo}, mathematical problem solving~\cite{drori2022neural}, and biomedical variant classification~\cite{li2025text2translation}. 
To run LLM-based services at scale, operators require large GPU clusters to support high-throughput inference workloads \cite{miao2023efficient_llm_serving,cai2025graphrepresentationbasedmodelpoisoning,linzheng3}. To this end, LLMs are typically deployed in the cloud, where they receive users’ full prompts and generate responses in real time. This cloud-based deployment model is driven by the massive size and computation demands of state-of-the-art LLMs, which exceed the resource limits of local devices.
%

%
%
%
%
%
Cloud-centric LLM deployment provides global accessibility and facilitates centralized model updates \cite{li2024llm_inference_serving,linzheng2}. However, it also introduces significant communication overhead and raises privacy risks due to the need to transmit full user prompts over the network \cite{linzheng2}.
These challenges are especially pronounced in privacy-critical domains such as medical, finance, and personalized services, where user prompts often contain sensitive personal information.
Unlike conventional software input, the prompts in these sensitive domains submitted to LLM often contain a rich semantic context that implicitly reveals the personal information, intentions, and preferences of users \cite{staab2024beyond,linzheng4}.
Consider the example in the medical domain: \emph{“I tested positive for HIV last week and have been experiencing fever and diarrhea. Should I be concerned about secondary infections?”} \cite{zeng2025privacyrestore}.
This prompt includes protected health information (PHI) , if intercepted or mishandled, leading to severe privacy violations. 
Further, not all prompts carry the same privacy sensitivity: simple queries such as \emph{“What is the capital of France?”} pose minimal risk and could be processed safely in the cloud \cite{pmlr-v119-acharya20a}. 
These limitations highlight the need for new inference paradigms that preserve privacy without sacrificing efficiency, especially in semantically complex and sensitive domains.

%
%
To mitigate the privacy risks inherent in cloud-centric inference, recent cloud–edge collaborative frameworks attempt to retain sensitive information on the edge devices \cite{jin2024cecollm, CE-LSLM,PICE}. 
A common design pattern is to perform inference locally on a Small Language Model (SLM) for privacy-critical prompts, while offloading non-sensitive queries to the cloud-hosted LLM based on binary routing mechanisms \cite{li2025edgecloudsurvey}. 
%
Alternatively, some systems adopt an encryption-inspired strategy: they apply local differential privacy (LDP) noise to the user’s full prompt before transmitting it to the cloud, and then rely on the edge to reconstruct or refine the cloud’s response \cite{lin2025emojiprompt,mai2024split}.
Although these approaches offer basic privacy protection, they remain fundamentally coarse-grained. 
First, binary routing decisions, based on simple thresholding of risk scores, can misclassify instructions, either overburdening the edge device or compromising privacy~\cite{qu2024mobile}. 
Second, prompt-level LDP schemes, such as Split-and-Denoise~\cite{mai2024split} and DP-Forward~\cite{du2023dpforward} uniformly perturb all tokens or embedding dimensions, regardless of their actual sensitivity.  
This uniform treatment leads to unnecessary utility degradation, especially for benign queries. Furthermore, when noise is applied indiscriminately, the cloud model receives semantically valid but semantically distorted commands, often generating vague, generic, or evasive responses (for example, 
'I cannot provide information about [$\mathtt{MASKED\_ENTITY}$').
In short, existing architectures lack the ability to tailor privacy mechanisms to the specific semantics of each prompt. This inefficiency motivates a more adaptive approach, one that dynamically selects privacy pathways based on contextual cues within the prompt itself.

%
%
To this end, we introduce Privacy-aware Routing for Inference with Semantic Modulation (PRISM), a cloud–edge collaborative framework that adaptively selects inference pathways based on prompt semantics and contextual privacy risk, as illustrated in Figure~\ref{fig:enter-label}.
%
%
Applying the same protection strategy to all inputs, PRISM employs a soft gating mechanism on the edge device to assess each user prompt and route it to one of three execution modes: (1) direct cloud inference for low-risk prompts; (2) sketch-based collaboration for moderately sensitive prompts; and (3) fully local generation for high-risk, privacy-critical queries.
The mode decision is made by a logit-based soft classifier that combines named entity recognition (NER), contextual cues (e.g. first-person references, entity types) and a statistical risk score. 
After the gating module, the system follows one of three execution paths: direct cloud inference and edge-only generation proceed immediately using the corresponding model, while cloud–edge collaborative mode triggers additional privacy-preserving procedures. Specifically, prompts in this mode are first obfuscated via a two-layer local differential privacy mechanism, then processed by the cloud to generate a semantic sketch, which is subsequently refined on the edge device to reconstruct a coherent and privacy-preserving response.
%
%

 To this end, we makes the following key contributions:
\begin{itemize}
    \item We propose PRISM, a novel privacy-aware routing framework for adaptive cloud–edge LLM inference. PRISM combines a context-aware gating mechanism with a three-mode execution pipeline (cloud, edge, and cloud–edge collaboration). In the collaborative mode, PRISM introduces a semantic modulation strategy that integrates adaptive two-layer local differential privacy with cloud-side sketch generation and edge-side refinement, enabling fine-grained, entity-level privacy control. 
    \item  To support this framework, we construct a synthetic and context-rich instruction dataset covering four domains: medical, tourism, banking, and general knowledge. This design reflects diverse real-world use cases and allows evaluation under varying privacy demands.
    \item   We implement PRISM on a real-world cloud–edge platform and evaluate its performance under varying privacy budgets, systematically testing multiple combinations of edge-side SLMs and cloud-hosted LLMs. Our results demonstrate that PRISM adapts to heterogeneous deployments, consistently achieving better inference quality, lower latency, and reduced edge-side energy consumption compared to uniform protection baselines.
\end{itemize}

\section{Framework Design}
This section details the design of PRISM, our privacy-aware cloud–edge inference framework. The system consists of three components: a user, a local edge device, and a remote cloud server. The edge device hosts an SLM, while the cloud server provides access to an LLM.
Inference begins when a user sends a prompt to the edge device, where two modules, Sensitivity Profiling and Soft Gating, analyze privacy risk and select an execution path: cloud, edge, or collaboration. For collaborative cases, PRISM applies Adaptive Two-Layer LDP to perturb sensitive entities and uses Cloud-Edge Semantic Sketch Collaboration to generate an abstract response on the cloud, which is then refined locally and returned to the user. 


\subsection{Sensitivity Profiling for Context-Aware Routing}
We design a lightweight edge-side module to assess the privacy sensitivity of user prompts before routing. Given a prompt \(P = \{x_1, x_2, \dots, x_n\}\) consisting of \(n\) tokens, the module first extracts a set of \(m\) named entities \(E = \{e_1, e_2, \dots, e_m\}\) using a named entity recognition (NER) engine. It then produces two outputs: (1) a scalar risk score \(R(P)\) reflecting the overall privacy sensitivity of the prompt, and (2) a binary mask \(\mathbf{d} \in \{0,1\}^m\) indicating which entities require protection.

We first apply a fast named entity recognition (NER) engine (e.g., Presidio Analyzer) to extract entities \(E = \{e_1, e_2, \dots, e_m\}\), where each entity \(e_i\) has an associated category label \(c_i \in \mathcal{C}\), and we assign a predefined sensitivity weight \(w_{c_i} \in [0,1]\) based on its category.
 These weights reflect the relevance for privacy of each category of entity (for example \(w_{\texttt{PERSON}} > w_{\texttt{NATIONALITY}}\)).

We define a  risk score for the prompt:
\begin{equation}
    R(P) = \sum_{i=1}^{m} w_{c_i} \cdot \mathbb{I}(e_i),
\end{equation}
where \(\mathbb{I}(e_i)\) indicates whether the entity is present in the current input.

To incorporate contextual signals, we define a binary indicator \(\Delta\) that activates when any private linguistic cue is detected:
\begin{equation}
    \Delta = \max_{x_j \in P} \mathbb{I}(x_j \in \mathcal{F}),
\end{equation}
where the private context set \(\mathcal{F}\) includes both first-person pronouns and detected entities of type \texttt{PERSON}.
Here, \(\mathbb{I}[\cdot]\) denotes the indicator function, and the condition evaluates whether any token in the prompt \(P\) overlaps with a predefined set of first-person pronouns or previously detected person-type entities.

Each entity \(e_i\) is conservatively flagged for protection if \(\Delta > 0\), producing the binary mask:
\begin{equation}
d_i = 
\begin{cases}
1, & \text{if } \Delta > 0, \\
0, & \text{otherwise}.
\end{cases} , \forall i \in \{1,2,..., m\}.
\end{equation}

This module captures both statistical and contextual signals cues on the edge devices. For instance, in the prompts (1)\emph{``I plan to travel solo to Tokyo for three days''} and (2)\emph{``Which country is Tokyo located in?''}, the entity \emph{Tokyo} appears in both but only the former implies private user intent due to the pronoun \emph{I}. These signals are forwarded to the later routing controller for decision-making. 



\begin{figure}[t]
    \centering
    \includegraphics[width=0.87\linewidth]{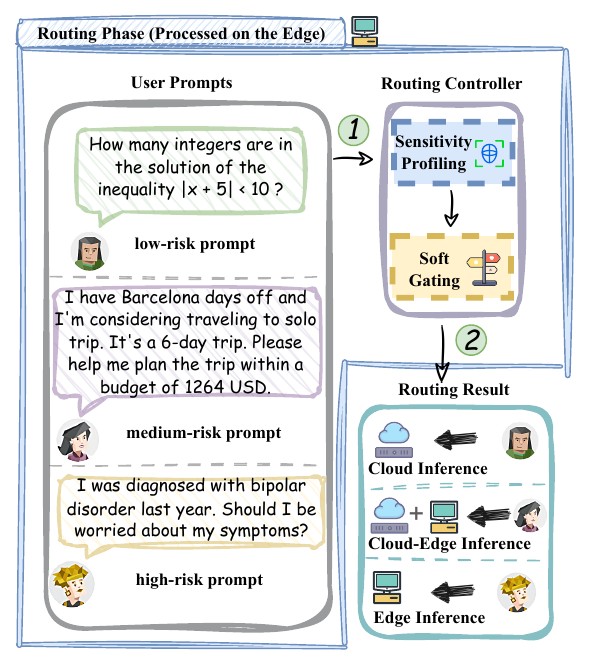}
    \caption{Illustration of the PRISM routing phase, where the edge-side controller analyzes prompt sensitivity and softly routes requests across cloud, edge, or collaborative paths.
}
    \label{fig:soft_gating}
\end{figure}

\subsection{Soft Gating with Entropy-Regularized Routing}
\label{sec:gating}

To adaptively balance privacy, utility, and energy cost, we propose a soft gating mechanism that maps sensitivity indicators to a probability distribution over three routing paths. Rather than committing to a hard decision boundary, this approach enables nuanced, context-aware execution by producing soft routing scores.
The gating module receives a feature vector $\mathbf{z} \in \mathbb{R}^{1+m}$ derived from the sensitivity profiling subsection, including the risk score $ R(P)$ and the sensitivity mask $\mathbf{d}$. These features are passed into a light-weight linear transformation \(f_\theta(\cdot)\) followed by softmax normalization:
\begin{equation}
    \boldsymbol{\pi} = \mathrm{softmax}(f_\theta(\mathbf{z})) \in \mathbb{R}^3.
\end{equation}
This yields a routing distribution \(\boldsymbol{\pi} = (\pi_{\text{cloud}}, \pi_{\text{collab}}, \pi_{\text{local}})\) over the three execution strategies: direct cloud inference, cloud–edge collaboration, and complete edge-side generation.
To encourage confident routing while preserving flexibility, we introduce an entropy penalty over the output distribution:
\begin{equation}
    \mathcal{L}_{\text{gating}} = \mathcal{L}_{\text{task}} + \lambda \cdot \mathcal{H}(\boldsymbol{\pi}),
\end{equation}
where  \(\mathcal{L}_{\text{task}}\) is the downstream task loss (cross-entropy loss for generation), $\mathcal{H}(\boldsymbol{\pi}) = -\sum_j \pi_j \log \pi_j$ is the entropy of the soft routing scores, and $\lambda$ is a tunable hyperparameter. Lower entropy encourages confident decisions, while higher entropy accommodates uncertain or ambiguous cases.

Figure~\ref{fig:soft_gating} provides an illustrative overview of the soft gating mechanism, showing how prompts of varying sensitivity are analyzed and softly routed to the most appropriate execution pathway based on their privacy characteristics.

At inference time, we select the most probable path by taking the top-1 decision from the soft routing distribution \(\boldsymbol{\pi}\), i.e., \(\arg\max_j \pi_j\). This deterministic selection ensures consistent privacy guarantees and avoids routing sensitive prompts to lower-protection paths due to randomness. Our formulation thus provides a principled and interpretable interface for integrating symbolic sensitivity signals with soft decision boundaries.

\subsection{Adaptive Two‐Layer Local Differential Privacy}
In cloud–edge collaborative collaboration, prompts containing sensitive entities must be encrypted before transmission.
Naive anonymization strategies, such as replacing a name with a generic token \(\langle\mathtt{NAME}\rangle\), do not mitigate the risk of linkage attacks. 
For example, consider the two records: (1)\emph{“\(\langle\mathtt{NAME}\rangle\) owns a black dog and often goes for a walk after dinner.”} and (2)\emph{“The owner of the black dog throws rubbish on Tuesdays.”} 
Although the user’s name \emph{“Bob”} has been masked, the shared semantic content \emph{“Black Dog"} allows an adversary to correlate the records and re-identify Bob. 
Worse, this cross-record inference also reveals behavioral details, such as Bob's routine and habits, amplifying the risk of privacy beyond identity disclosure.

Similarly, applying uniform \(\varepsilon\)-LDP, either throughout the prompt or uniformly across all types of entities, introduces suboptimal trade-offs: 
the former may disrupt linguistic coherence, while the latter either underprotects sensitive entities or overly perturbs benign ones, affecting downstream utility.
Moreover, when noise is applied indiscriminately, the cloud model receives syntactically valid but semantically distorted prompts, often yielding vague or evasive completions (e.g., \emph{``I cannot provide information about \(\langle\mathtt{MASKED\_ENTITY}\rangle\)''}).

To address these issues, we propose a \textit{two-layer adaptive LDP mechanism} that separately perturbs each entity’s \textit{category} and \textit{value} using independently allocated budgets \(\varepsilon_1\) and \(\varepsilon_2\), such that the total budget is preserved: \(\varepsilon_{\text{total}} = \varepsilon_1 + \varepsilon_2\). The allocation is determined by the sensitivity weight \(w_{c_i}\), which reflects how sensitive the category \(c_i\) of entity \(e_i\) is. This weight comes from our profiling module.

\begin{equation}
\begin{aligned}
\varepsilon_1 &= \varepsilon_{\text{total}} \cdot \frac{w_{c_i}}{w_{c_i} + (1 - w_{c_i}) \cdot \alpha}, \\
\varepsilon_2 &= \varepsilon_{\text{total}} - \varepsilon_1 ,
\end{aligned}
\end{equation}

where \(\alpha \in (0,1]\) is a tunable hyperparameter. This formulation ensures high-sensitivity categories (e.g., \emph{NAME}) receive more category-level protection, while lower-sensitivity types (e.g., \emph{NATIONALITY}) allocate more to value-level obfuscation. This adaptive allocation supports flexible trade-offs without violating $\varepsilon$-LDP guarantees, following composition rules in local privacy \cite{duchi2013local}.
When \(w_{c_i}\) is high (e.g., \emph{NAME, ID, DIAGNOSIS}), more budget is allocated to \(\varepsilon_1\) to protect the entity type, reducing the risk of inferring sensitive categories. Conversely, for lower \(w_{c_i}\) (e.g., \emph{ORGANIZATION, NATIONALITY, LOCATION}), the mechanism allocates more to \(\varepsilon_2\) to preserve semantic utility while anonymizing the value. 

Both the category-layer and value-layer perturbation follow the Randomized Response (RR) mechanism, a canonical construction for achieving $\varepsilon\text{-LDP:}$ in finite discrete domains~\cite{wang2017locally}.

\begin{equation}
\begin{aligned}
&\textit{(1) Category-Layer } \varepsilon_1\text{-LDP:} \\
&p_1 = \frac{\exp(\varepsilon_1)}{\exp(\varepsilon_1)+K_1-1}, \quad
c_i^* =
\begin{cases}
c_i & \text{w.p. } p_1, \\
\neq c_i & \text{w.p. } \frac{1 - p_1}{K_1 - 1}.
\end{cases} \\[1ex]
&\textit{(2) Value-Layer } \varepsilon_2\text{-LDP:} \\
&p_2 = \frac{\exp(\varepsilon_2)}{\exp(\varepsilon_2)+K_2-1}, \quad
e_i^* =
\begin{cases}
e_i & \text{w.p. } p_2, \\
\neq e_i & \text{w.p. } \frac{1 - p_2}{K_2 - 1}.
\end{cases}
\end{aligned}
\end{equation}

Figure~\ref{fig:two_layer_ldp} illustrates the adaptive two-layer LDP workflow. After profiling, each entity undergoes privacy budget allocation based on its category sensitivity weight \(w_{c_i}\). The entity is then sequentially processed by two perturbation layers: the first samples a noisy category label from the set \(\mathcal{C}\) of size \(K_1\) with probability $p_1$, and the second samples a value from the corresponding value set \(\mathcal{V}_{c_i}\) of size \(K_2\) with probability $p_2$.
To illustrate the adaptive nature of the method, we show two representative cases. Type B entities (e.g., \emph{NAME}) are assigned high sensitivity weights and therefore undergo heavier category-level obfuscation to hide their semantic type. In contrast, type A entities (e.g., \emph{LOCATION}) are considered less sensitive; they retain their category label while receiving stronger perturbation on the value level, ensuring semantic coherence of the prompt. This selective strategy enables better privacy–utility trade-offs, especially in resource-constrained cloud–edge deployments.

\begin{figure}[t]
    \centering
    \includegraphics[width=0.92\linewidth]{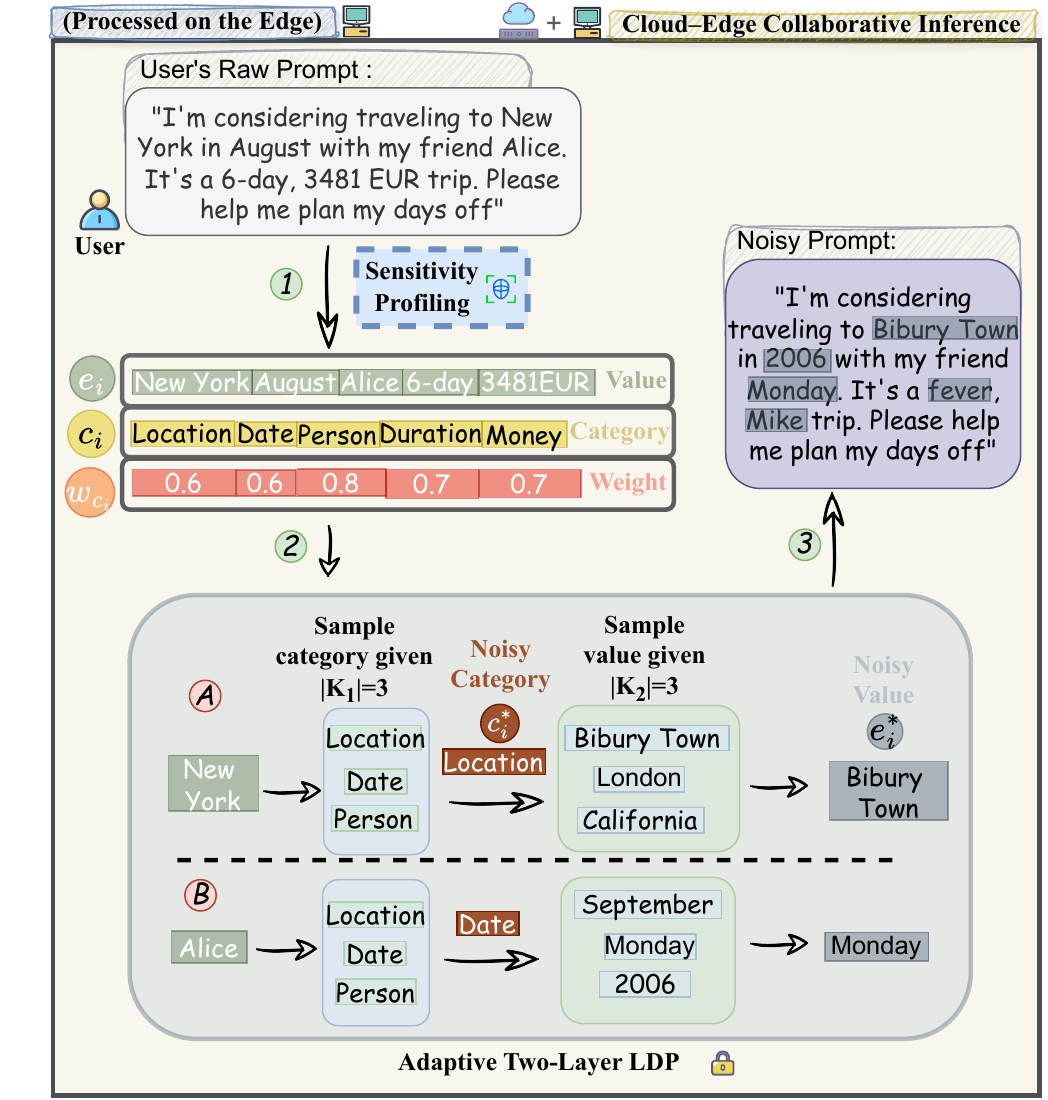}
    \caption{Adaptive Two-Layer LDP for Entity Obfuscation}
    \label{fig:two_layer_ldp}
\end{figure}

Our method matches the optimal $\varepsilon$-LDP design for finite discrete domains \cite{ wang2017locally}, but extends it via category-aware allocation. When \(w_{c_i} \to 1\), \(\varepsilon_1 \to \varepsilon_{\text{total}}\), maximizing protection of the category. When \(w_{c_i} \to 0\), the reverse holds, preserving utility. Compared to static masking or uniform LDP, our risk-adaptive mechanism preserves semantic coherence by selectively perturbing high-risk entities more aggressively while minimizing distortion on benign ones. 


\subsection{Cloud-Edge Semantic Sketch Collaboration}
\label{sec:cloud_edge_sketch}

In our architecture, each input prompt first passes through a routing phase that determines the appropriate execution mode: edge-only, cloud-edge collaboration, or cloud-only, based on its assessed privacy risk and resource constraints. Prompts routed to the collaborative mode then undergo our adaptive two-layer LDP mechanism,producing a perturbed prompt \(  P^* \), where \( P^* \) replaces sensitive entities in \(P\) using randomized response. The resulting noisy prompt is directly transmitted to the cloud in plain text, avoiding the transmission of embeddings and eliminating the need for synchronized tokenizers or shared embedding spaces between devices.

Given the perturbed prompt \( P^* \), the cloud-side LLM \(\mathcal{G}_{\text{cloud}}\) generates a semantic sketch \( S \in \mathcal{T}_{\text{sketch}} \) using a few-shot in-context prompt using a demonstration set \(\mathcal{D}_{\text{cloud}} = \{(P^{*(l)}, S^{(l)})\}_{l=1}^{k}\).  
The generation context is constructed as \(\mathcal{C}_{\text{cloud}} = [\mathcal{D}_{\text{cloud}}, (P^*,\_)]\),
where \(\_\) indicates the sketch to be generated. The cloud model then produces:
\begin{equation}
S = \mathcal{G}_{\text{cloud}}(\mathcal{C}_{\text{cloud}}).
\end{equation}

The sketch \(S\) adopts a concise, structured format and omits sensitive entities obfuscated in \(P^*\). This design ensures (i) \emph{semantic alignment}, preserving the original intent despite noise, and (ii) \emph{structural regularity}, supporting downstream integration.

Upon receiving the sketch \(S\), the edge-side SLM \(\mathcal{G}_{\text{edge}}\) reconstructs the final response \(\hat{R}\) by conditioning on both \(S\) and the original prompt \(P\), which remains locally available.

Similarly, we define the edge-side demonstration set as \(\mathcal{D}_{\text{edge}} = \{(P^{(l)}, S^{(l)}, R^{(l)})\}_{l=1}^{k}\), and construct the generation context as \(\mathcal{C}_{\text{edge}} = [\mathcal{D}_{\text{edge}}, (P, S, \_)]\), where \(\_\) denotes the final response to be generated. The final response is then generated by the edge model:
\begin{equation}
\hat{R} = \mathcal{G}_{\text{edge}}(\mathcal{C}_{\text{edge}}).
\end{equation}

This collaborative inference design preserves privacy while leveraging semantic abstraction for faithful and efficient response generation.

To unify the overall decision and execution process, we present the complete PRISM inference pipeline in Algorithm~\ref{alg:prism}. It integrates sensitivity-aware routing, entity-level perturbation via adaptive LDP, and collaborative generation with sketch-based semantic alignment. This algorithm formalizes the end-to-end behavior of our framework across all execution modes, including edge-only, cloud-only, and cloud–edge collaboration.

\begin{algorithm}[t]
\caption{PRISM Framework}
\label{alg:prism}
\KwIn{
User prompt \(P = \{x_1, \dots, x_n\}\);\quad

Cloud demo set \(\mathcal{D}_{\text{cloud}} = \{(P^{*(l)}, S^{(l)})\}_{l=1}^{k}\);\quad

Edge demo set \(\mathcal{D}_{\text{edge}} = \{(P^{(l)}, S^{(l)}, R^{(l)})\}_{l=1}^{k}\);\quad
Models \(\mathcal{G}_{\text{cloud}}, \mathcal{G}_{\text{edge}}\)
}

\KwOut{Final response \(\hat{R}\)}

\((R(P), \mathbf{d}) \leftarrow \textsc{SensitivityProfiling}(P)\)\;
\(\boldsymbol{\pi} \leftarrow \textsc{SoftGating}([R(P), \mathbf{d}])\)\;
\(\texttt{mode} \leftarrow \arg\max_j \pi_j\)\;

\If{\texttt{mode} $=$ \texttt{edge-only}}{
    \Return \( \hat{R} = \mathcal{G}_{\text{edge}}(P) \)
}
\If{\texttt{mode} $=$ \texttt{cloud-only}}{
    \Return \( \hat{R} = \mathcal{G}_{\text{cloud}}(P) \)
}

    
    
    

\If{\texttt{mode} $=$ \texttt{collaborative}}{
    \(P^* \leftarrow \textsc{TwoLayerLDP}(P, \mathbf{d})\)\;
    \(\mathcal{C}_{\text{cloud}} \leftarrow [\mathcal{D}_{\text{cloud}}, (P^*, \_)]\)\;
    \(S \leftarrow \mathcal{G}_{\text{cloud}}(\mathcal{C}_{\text{cloud}})\)\;

    \(\mathcal{C}_{\text{edge}} \leftarrow [\mathcal{D}_{\text{edge}}, (P, S, \_)]\)\;
    \(\hat{R} \leftarrow \mathcal{G}_{\text{edge}}(\mathcal{C}_{\text{edge}})\)\;
    \Return \(\hat{R}\)
}

\end{algorithm}

\section{Analysis}



In this section, we provide a rigorous analysis of the privacy properties of our selective-privacy cloud–edge inference framework. Specifically, we establish the local differential privacy guarantees of our two-layer perturbation mechanism and characterize how sensitivity weights influence budget allocation. 

\begin{theorem}[Two-Layer LDP Privacy Guarantee]
Let \( M \) be the adaptive two-layer mechanism applied to a sensitive entity \( e_i \) with category \( c_i \in \mathcal{C} \) (of size \( K_1 \)) and value domain \( \mathcal{V}_{c_i^*} \) (of size \( K_2 \)). The mechanism sequentially applies:
\begin{enumerate}
    \item Category-level randomized response \(M_1\) with budget \(\epsilon_1\), producing \(c_i^*\);
    \item Value-level randomized response \(M_2\) with budget \(\epsilon_2\), producing \(e_i^*\).
\end{enumerate}
Then \(M = M_2 \circ M_1\) satisfies \((\epsilon_1 + \epsilon_2)\)-local differential privacy over the pair \((c_i, e_i)\).
\end{theorem}

\begin{proof}
According to the definition of \(\varepsilon\)-local differential privacy\cite{Dagan2006PASCAL_RTE,Dwork2006DifferentialPrivacy}, a mechanism \(M\) satisfies \(\varepsilon\)-LDP if for any two distinct inputs \(x, x'\) and any output \(y\), we have:
\[
\frac{\Pr[M(x) = y]}{\Pr[M(x') = y]} \leq \exp(\varepsilon).
\]
We now verify this condition for each component.

\emph{Step 1: Category-level privacy.}  
We apply a standard randomized response mechanism over the finite category domain \(\mathcal{C}\) of size $K_1$, where the probability of outputting the true class is
\(
p_1 = \frac{\exp(\varepsilon_1)}{\exp(\varepsilon_1)+K_1-1},
\)
and the remaining probability mass \((1 - p_1)\) is distributed uniformly over the \(K_1 - 1\) other categories:
\[
\Pr[M_1(c_i) = c_i^*] =
\begin{cases}
p_1, & \text{if } c_i^* = c_i, \\
\frac{1 - p_1}{K_1 - 1}, & \text{if } c_i^* \neq c_i.
\end{cases}
\]

For any two distinct inputs \(c_i, c_i' \in \mathcal{C}\), and any output \(c_i^*\), we compute the likelihood ratio:

\[
\frac{\Pr[M_1(c_i) = c_i^*]}{\Pr[M_1(c_i') = c_i^*]}=
\begin{cases}
    \frac{p_1}{\frac{1 - p_1}{K_1 - 1}} = \exp(\varepsilon_1), & \text{if } c_i^* = c_i, \\
     \frac{\frac{1 - p_1}{K_1 - 1}}{p_1} = \frac{1}{\exp(\varepsilon_1)}, & \text{if } c_i^* = c_i'.
\end{cases}
\]

In either case, the ratio is bounded:
\(
\frac{\Pr[M_1(c_i) = c_i^*]}{\Pr[M_1(c_i') = c_i^*]} \leq \exp(\epsilon_1).
\)
Thus, \(M_1\) satisfies \(\varepsilon_1\)-LDP.

\emph{Step 2: Value-level privacy.}  
Over the value domain \(\mathcal{V}_{c_i^*}\) of size \(K_2\), the randomized response mechanism \(M_2\) outputs:
\[
\Pr[M_2(e_i) = e_i^*] =
\begin{cases}
    p_2 = \frac{\exp(\varepsilon_2)}{\exp(\varepsilon_2)+K_2-1}, & \text{if } e_i^* = e_i, \\
    \frac{1 - p_2}{K_2 - 1}, & \text{if } e_i^* \neq e_i.
\end{cases}
\]

A similar analysis shows that for any two distinct values \(e_i, e_i'\), the likelihood ratio is bounded by \(\exp(\varepsilon_2)\), so \(M_2\) satisfies \(\varepsilon_2\)-LDP.
\emph{Step 3: Composition.}  
Since \(M_1\) and \(M_2\) are applied independently and sequentially to disjoint components \((c_i, e_i)\), we apply the sequential composition theorem~\cite{duchi2013local} to obtain:
\[
M(c_i, e_i) = M_2 \circ M_1 \Rightarrow (\varepsilon_1 + \varepsilon_2)\text{-LDP}.
\]
\end{proof}

This theorem guarantees that for each individually selected sensitive entity, our adaptive two-layer LDP provides bounded privacy leakage consistent with the allocated total budget \(\epsilon_{\text{total}} = \epsilon_1 + \epsilon_2\). 

\begin{theorem}[Effect of Sensitivity Weight on Budget Allocation]
Let \( w_{c_i} \in [0,1] \) be the sensitivity weight of category \( c_i \), and let \( \varepsilon_{\text{total}} > 0 \) be the total privacy budget. Define the allocation as:
\[
\varepsilon_1 = \varepsilon_{\text{total}} \cdot \frac{w_{c_i}}{w_{c_i} + (1 - w_{c_i}) \cdot \alpha}, \quad
\varepsilon_2 = \varepsilon_{\text{total}} - \varepsilon_1,
\]
where \( \alpha \in (0,1] \) is a tunable hyperparameter that controls the relative importance of value protection.
Then:\begin{enumerate}
    \item When \( w_{c_i} = 1 \), we have \( \varepsilon_1 = \varepsilon_{\text{total}} \), \( \varepsilon_2 = 0 \).
    \item When \( w_{c_i} = 0 \), we have \( \varepsilon_1 = 0 \), \( \varepsilon_2 = \varepsilon_{\text{total}} \).
    \item \( \varepsilon_1(w_{c_i}) \) is monotonically increasing on \([0,1]\); \( \varepsilon_2(w_{c_i}) \) decreases monotonically.
\end{enumerate}
\end{theorem}

\begin{proof}
We analyze the function:
\[
\varepsilon_1(w_{c_i}) = \varepsilon_{\text{total}} \cdot \frac{w_{c_i}}{w_{c_i} + (1 - w_{c_i}) \cdot \alpha}, \quad w_{c_i} \in [0,1].
\]

\emph{Boundary Cases:}
\begin{itemize}
    \item When \( w_{c_i} = 1 \), the fraction becomes \( \frac{1}{1} = 1 \), so
    \(
    \varepsilon_1 = \varepsilon_{\text{total}},  \varepsilon_2 = 0.
    \)
    \item When \( w_{c_i} = 0 \), the fraction becomes \( \frac{0}{\alpha} = 0 \), so
    \(
    \varepsilon_1 = 0, \varepsilon_2 = \varepsilon_{\text{total}}.
    \)
\end{itemize}

\emph{Monotonicity:} Consider the derivative of \( \varepsilon_1(w_{c_i}) \) with respect to \( w_{c_i} \):
\[
\frac{d\varepsilon_1(w_{c_i})}{dw_{c_i}} = \varepsilon_{\text{total}} \cdot \frac{\alpha}{\left((\alpha - 1)w_{c_i} - \alpha\right)^2}.
\]

Note that the denominator is a square term and therefore always nonnegative. It equals zero if and only if:
\[
(\alpha - 1)w_{c_i} - \alpha = 0 \quad \Leftrightarrow \quad w_{c_i} = \frac{\alpha}{\alpha - 1}.
\]

Since \(\alpha \in (0,1]\), we have \(\alpha - 1 < 0\), so \(\frac{\alpha}{\alpha - 1} < 0\), which lies outside the domain \(w_{c_i} \in [0,1]\). Therefore, the denominator is strictly positive throughout the valid domain.
Moreover, since the numerator is the product of \(\alpha \in (0,1]\) and $\varepsilon_{\text{total}} > 0 $, it is strictly positive. 

Hence, the entire expression is strictly greater than zero for all \(w_{c_i} \in [0,1]\), proving that the allocation function \(\varepsilon_1(w_{c_i})\) is monotonically increasing, and \( \varepsilon_2(w_{c_i}) = \varepsilon_{\text{total}} - \varepsilon_1(w_{c_i}) \) is strictly decreasing.

\end{proof}

This theorem confirms that higher category sensitivity weights allocate more budget to category-level perturbation (\(\varepsilon_1\)), as changing the category of an entity induces greater semantic distortion, offering stronger protection for highly sensitive categories.  Lower weights favor value-level noise (\(\varepsilon_2\)) to preserve utility for less sensitive entities.

\section{Evaluation}



\subsection{Evaluation Setting}

\begin{figure*}[t]
    \centering
    \begin{subfigure}[t]{0.32\linewidth} 
        \centering
        \includegraphics[width=\linewidth]{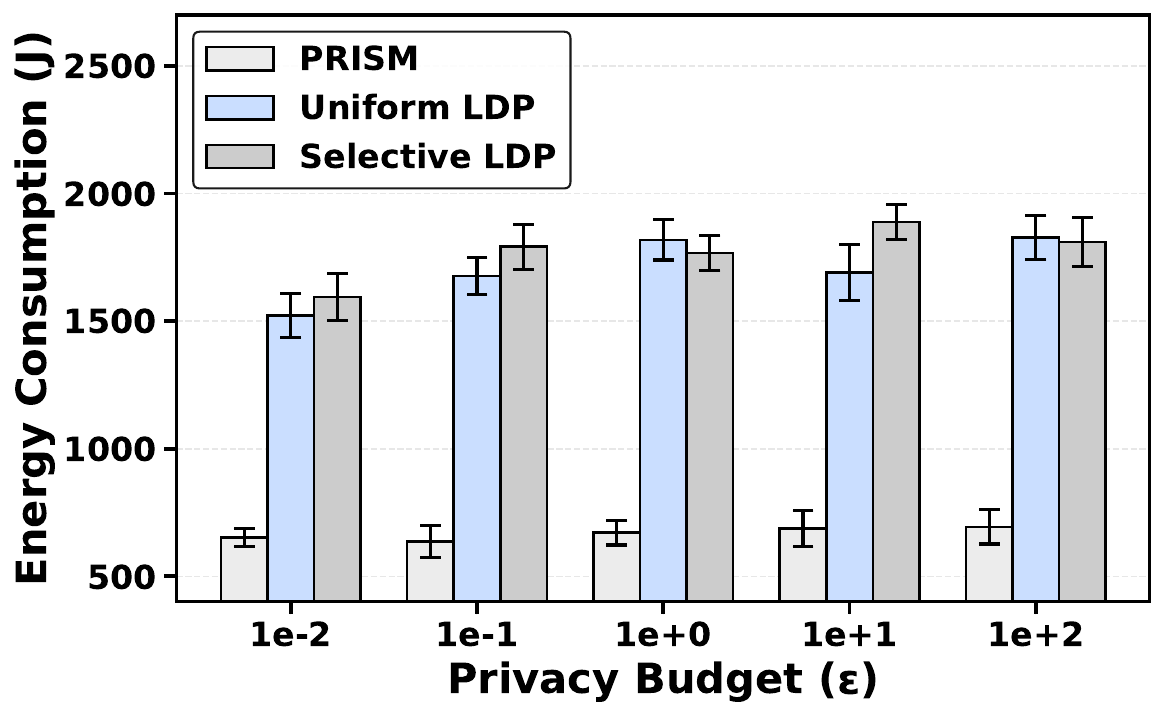}
        \caption{Energy Consumption}
        \label{fig:energy}
    \end{subfigure}
    \hfill
    \begin{subfigure}[t]{0.32\linewidth}
        \centering
        \includegraphics[width=\linewidth]{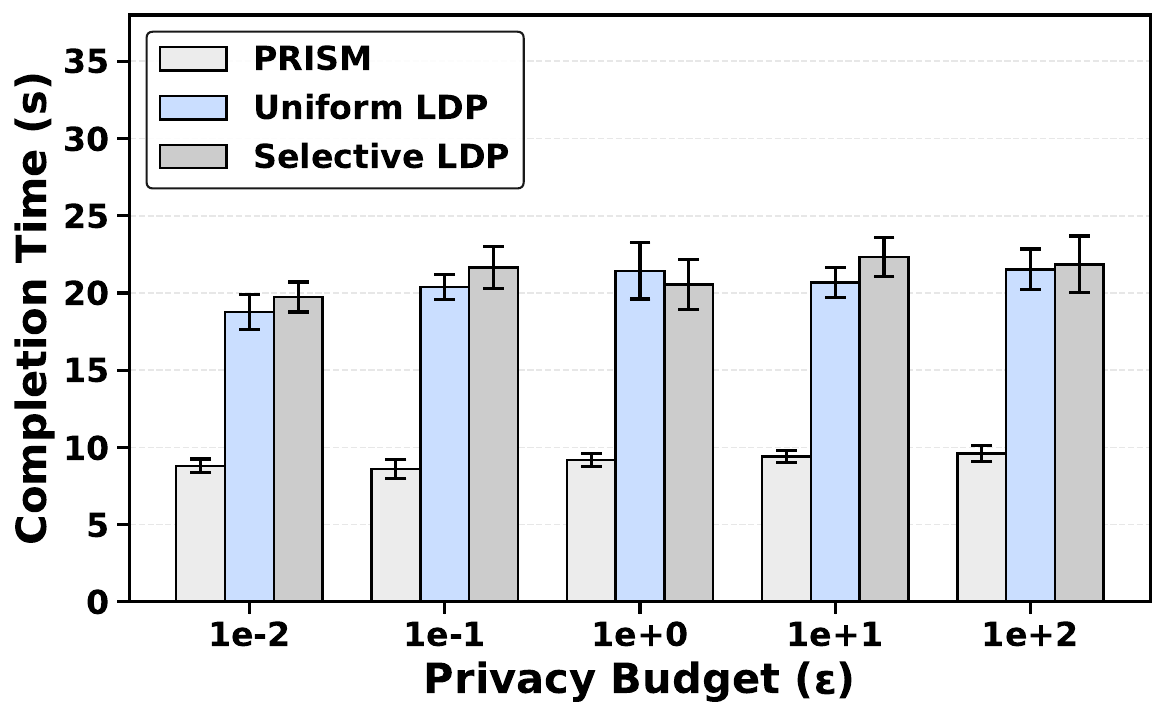}
        \caption{Completion Time}
        \label{fig:time}
    \end{subfigure}
    \hfill
    \begin{subfigure}[t]{0.32\linewidth}  
        \centering
        \includegraphics[width=\linewidth]{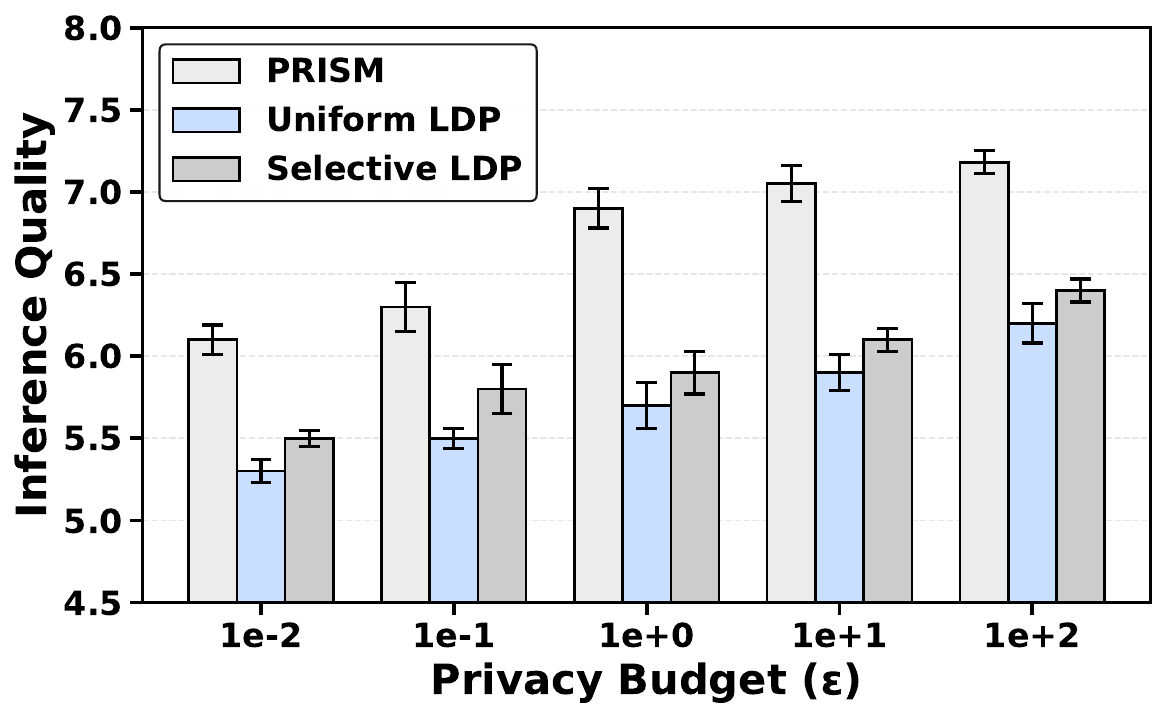}
        \caption{Inference Quality}
        \label{fig:llm}
    \end{subfigure}

    \caption{Comparison of privacy-preserving methods across privacy budgets on three dimensions: 
             (a) energy consumption in Joules, (b) completion time in seconds,  and (c) inference quality evaluated by LLM as a judge.} 
    \label{fig:privacy_all_metrics}
\end{figure*}

\emph{{Dataset: }}To facilitate comprehensive evaluation under privacy-sensitive scenarios, we construct a semi-synthetic instruction dataset designed to emulate realistic user–LLM interactions across four representative application domains: (1) \emph{Tourism planning}, covering user identities, travel budgets, and destinations; (2) \emph{Medical consultation}, involving demographic attributes and symptom descriptions, partially adapted from the dataset used in PrivacyRestore \cite{zeng2025privacyrestore}; (3) \emph{Banking services}, featuring transaction histories, account identifiers, and institutional details; and (4) \emph{General knowledge}, drawn from MT-Bench~\cite{zheng2023judging} as a non-sensitive dataset. Each domain contains 40 prompts with diverse phrasings and structured entity distributions, enabling controlled variation in contextual sensitivity. This dataset supports systematic, fine-grained benchmarking of privacy-aware LLM inference strategies.

\emph{Environment Setup: }
We evaluate PRISM in a realistic cloud-edge deployment to reflect practical resource constraints and latency characteristics.
The edge-side model is deployed on a local workstation equipped with a NVIDIA RTX 3070 laptop GPU running Windows 10. We use the \textit{llama\_cpp\_python} backend to serve quantized SLMs, with the GPU offloading configuration set to \textit{gpu\_layers=32}. 
The cloud-side model is accessed via API calls to hosted LLMs The cloud performs sketch generation based on the perturbed prompts received from the edge and returns the abstracted sketches for downstream reconstruction.

\emph{Baselines: }We evaluate our framework against the following representative baselines. (1) \emph{Uniform LDP} adds Laplace noise to all input tokens and sends the perturbed prompt to the cloud for inference, followed by edge-side refinement without semantic sketching \cite{mai2024split}. (2) \emph{Selective LDP} applies LDP only to entities identified by the NER model \cite{shi-etal-2022-selective}, then performs cloud inference and edge refinement similarly. (3)  \emph{Cloud-only} sends the entire unaltered prompt to a cloud-based LLM for inference. (4) \emph{Edge-only} performs inference locally on an edge device using SLM with the original prompt.

\emph{Evaluation Metrics: }We evaluate each method using three key metrics. (1) \emph{Inference Quality} is assessed by GPT-4o-based scoring on a scale from 1 to 10, reflecting the relevance, coherence, and informativeness of the generated responses \cite{zheng2023judging}. (2) \emph{Energy Consumption} is measured in Joules using Windows-based power monitoring tools, reflecting the total energy consumed by the edge device during the inference period. For cloud-only inference, this includes only the energy cost incurred by data transmission and idle system operations on the edge, since the actual model execution occurs remotely on the cloud. (3) \emph{ Completion Time} denotes the end-to-end latency from prompt input to final response, indicating system responsiveness.

\subsection{Results}

Table~\ref{tab:avg_method_summary} shows that PRISM achieves the best overall efficiency among privacy-preserving methods, with only 7.92\,s completion time and 687\,J energy consumption. This efficiency stems from its adaptive routing mechanism, which sends non-sensitive prompts directly to the cloud, routes low-risk ones to cloud–edge inference, and retains highly sensitive prompts for local processing. While Cloud-only attains the best raw performance (5.13\,s, 296\,J, IQ = 8.14), it offers no privacy protection. In contrast, PRISM ranks second across all methods, and requires only 1.54$\times$ the latency and 2.32$\times$ the energy of Cloud-only, making it a favorable privacy–efficiency trade-off. Meanwhile, Uniform and Selective LDP methods incur over 20\,s latency and 1700+\,J energy due to indiscriminate or entity-based noise injection.

Figure~\ref{fig:privacy_all_metrics} illustrates the performance trends of three privacy-preserving methods across varying privacy budgets. We observe that PRISM consistently outperforms Uniform and Selective LDP in all three metrics, achieving notably lower energy consumption (e.g., 652\,J vs.\ 1800+\,J), shorter completion time (e.g., 8.8\,s vs.\ 21+\,s), and higher inference quality (up to 7.2 vs.\ 6.2 or lower). PRISM exhibits  stable performance across privacy budgets. This robustness arises from its adaptive routing and minimal perturbation design, which avoids over-noising and computational waste while maintaining utility under strong privacy constraints.

\begin{table}
\centering
\small
\begin{tabular}{lrrr}
\toprule
    Method &   Ct.(s) &   Ec.(J) &   IQ. \\
\midrule
     PRISM &        7.92 &      687.16 &     6.88 \\
     Uniform LDP &        20.56 &      1707.6 &     5.72 \\
     Selective LDP &        21.22 &      1770.8 &     5.94 \\
 Edge-Only &       17.84 &     1573.9 &     5.09 \\
Cloud-Only &        \textbf{5.13} &      \textbf{296.27} &     \textbf{8.14} \\
\bottomrule
\end{tabular}
\caption{Performance metrics across methods. Metrics include completion time (Ct.), energy consumption (Ec.), and inference quality (IQ.). The best performance per metric is shown in \textbf{bold}. }
\label{tab:avg_method_summary}
\end{table}

Table~\ref{tab:main_results} presents the performance of PRISM in eight combinations of cloud-edge models. We observe that all pairings deliver high inference quality (IQ $\geq$ 6.9) with moderate latency (7.08–8.60\,s) and low energy (632–739\,J), demonstrating the flexibility of PRISM in heterogeneous deployments. In particular, the combination of GPT-4o (L1) and Qwen1.5-1.8B (S2) achieves the fastest inference (7.08\,s) and the lowest energy (632\,J), while GPT-4o with StableLM (S3) yields the highest quality (7.16). Meanwhile, the Qwen3-235B (L2) variants offer a slightly slower response but match or exceed the quality of the generation (up to 7.22), underscoring the adaptability of PRISM to varying cloud / edge model capabilities.

\begin{table}[h]
\small
\centering
\begin{tabular}{llll}  
  \toprule
  \multirow{2}{*}{Model Combinations}  & \multicolumn{3}{c}{PRISM}  \\
  \cmidrule(lr){2-4} 
  & Ct.(s) & Ec.(J) & IQ.  \\
    \midrule
     \emph{L1} \textbf{+}  \emph{S1}     &  8.29  &  683.83  &  7.00  \\
     \emph{L1} \textbf{+}  \emph{S2}      &  \textbf{7.08} &  \textbf{632.24} &  6.91  \\
     \emph{L1} \textbf{+}  \emph{S3}     &  7.34  &  657.88  &  7.16 \\
     \emph{L1} \textbf{+}  \emph{S4}    & 7.35   &  653.62  &  5.28   \\
     \emph{L2} \textbf{+}  \emph{S1}    &  8.59  &  738.88  &  \textbf{7.22}  \\
     \emph{L2} \textbf{+}  \emph{S2}      &  8.60  &  739.59  &  7.06     \\
     \emph{L2} \textbf{+}  \emph{S3}       &  8.00  &  693.13 &  7.19   \\
     \emph{L2} \textbf{+}  \emph{S4}      &  8.11 &  698.10  &  7.19   \\
     \bottomrule
  \end{tabular}
  \caption{
Performance of PRISM across various combinations of cloud-side LLMs (L1:GPT-4o, L2:Qwen3-235B) and edge-side SLMs (S1:Phi-3.5-mini-3.5B, S2:Qwen1.5-1.8B, S3:StableLM-2-Zephyr-1.6B, S4:TinyLLaMA-1.1B).
}
\label{tab:main_results}
\end{table}

\section{Conclusion}
This work presents PRISM, a novel privacy-aware cloud–edge inference framework that dynamically routes user prompts based on semantic sensitivity. By integrating an adaptive two-layer local differential privacy mechanism, and a semantic sketch-based collaboration protocol, PRISM enables efficient, privacy-preserving inference without compromising utility. Comprehensive evaluations across four domains and eight model configurations demonstrate that PRISM achieves superior performance trade-offs, incurring only 1.54× latency and 2.32× energy overhead compared to cloud-only inference, while preserving strong privacy guarantees and maintaining high output quality.

As future work, we plan to extend PRISM to support collaborative inference across multiple edge devices, each equipped with an SLM. This setting introduces new challenges in routing, load balancing, and cloud-edge-device coordination, which we aim to address through decentralized scheduling and federated optimization mechanisms.

\bibliography{r.bib}

\newpage
\setlength{\leftmargini}{20pt}
\makeatletter\def\@listi{\leftmargin\leftmargini \topsep .5em \parsep .5em \itemsep .5em}
\def\@listii{\leftmargin\leftmarginii \labelwidth\leftmarginii \advance\labelwidth-\labelsep \topsep .4em \parsep .4em \itemsep .4em}
\def\@listiii{\leftmargin\leftmarginiii \labelwidth\leftmarginiii \advance\labelwidth-\labelsep \topsep .4em \parsep .4em \itemsep .4em}\makeatother

\setcounter{secnumdepth}{0}
\renewcommand\thesubsection{\arabic{subsection}}
\renewcommand\labelenumi{\thesubsection.\arabic{enumi}}

\newcounter{checksubsection}
\newcounter{checkitem}[checksubsection]

\newcommand{\checksubsection}[1]{%
  \refstepcounter{checksubsection}%
  \paragraph{\arabic{checksubsection}. #1}%
  \setcounter{checkitem}{0}%
}

\newcommand{\checkitem}{%
  \refstepcounter{checkitem}%
  \item[\arabic{checksubsection}.\arabic{checkitem}.]%
}
\newcommand{\question}[2]{\normalcolor\checkitem #1 #2 \color{blue}}
\newcommand{\ifyespoints}[1]{\makebox[0pt][l]{\hspace{-15pt}\normalcolor #1}}



\section*{Supplementary material}

# Appendix for
# PRISM: Privacy-Aware Routing for Adaptive Cloud–Edge LLM Inference via Semantic Sketch Collaboration

## Appendix

This appendix provides additional theoretical analyses, implementation details, and experimental results to support the findings presented in the main paper. The contents are organized as follows:

- **Appendix A:** PRISM Architecture and Model Settings
- **Appendix B:** Theoretical Extensions and Proof Sketches

## A  PRISM Architecture and Model Settings

### A.1  Overview

The PRISM framework consists of four primary components deployed across edge and cloud environments:

- **Edge-side Entity Profiler**: Performs fine-grained sensitivity detection over user prompts using a lightweight named entity recognition (NER) module.
- **Soft Routing Gater**: A lightweight neural gating module that maps symbolic sensitivity features to a soft probability distribution over three execution paths (cloud-only, edge-only, and collaboration). It is trained with entropy regularization to balance privacy risk, utility, and energy efficiency.
- **Adaptive Two-Layer LDP Module**: Applies hierarchical perturbations over sensitive entities based on entity type and content value.
- **Cloud–Edge Semantic Collaboration Module**: Generates abstract sketches on the cloud-hosted LLM, which are then decoded and refined on the edge-side small language model (SLM).

The architecture supports flexible configurations of SLMs and LLMs, enabling deployment across a range of hardware environments.

### A.2  Small Language Models (SLMs)

We evaluate PRISM with multiple lightweight SLMs deployed on edge devices. All models are loaded in GGUF format and quantized to reduce memory and latency overhead. The following models are used in our experiments:

| Model Name | Parameters | Quantization | File |
|---|---|---|---|
| TinyLLaMA Chat | 1.1B | Q8_0 | `tinyllama-1.1b-chat-v1.0.Q8_0.gguf` |
| Phi-3.5 Mini Instruct | 3.8B | Q6_K_L | `Phi-3.5-mini-instruct-Q6_K_L.gguf` |
| StableLM Zephyr | 1.6B | Q6_K | `stablelm-2-zephyr-1_6b.Q6_K.gguf` |
| Qwen1.5 Chat | 1.8B | Q6_K | `qwen1_5-1_8b-chat-q6_k.gguf` |

Table 1: Edge-deployed SLMs used in PRISM evaluation.

All SLMs are run on consumer-grade GPUs (e.g., NVIDIA RTX 3070 laptop GPU) with `llama.cpp`-based inference. For latency and energy profiling, models are launched in a single-batch, no-streaming configuration.

### A.3 Cloud-side LLM

We use GPT-4o (OpenAI, 2024) as the backbone for cloud-side inference and sketch generation. All API-based interactions follow OpenAI's default chat-completion endpoint with system-level few-shot prompting. Model temperature is set to 0.7 unless otherwise specified.

### A.4 Soft Gating with Entropy-Regularized Routing

The PRISM framework employs a soft gating module to compute context-aware routing distributions over three execution paths: cloud-only, edge-only, and cloud–edge collaboration.

Given a prompt $P = \{x_1, x_2, \ldots, x_n\}$, we extract a set of named entities $E = \{e_1, \ldots, e_m\}$, each associated with a binary sensitivity indicator $d_i \in \{0, 1\}$. We also compute an overall scalar risk score $R(P) \in \mathbb{R}$. These are concatenated to form the input feature vector $\mathbf{z} = [R(P); \mathbf{d}] \in \mathbb{R}^{1+m}$. For consistency, the sensitivity mask $\mathbf{d}$ is padded or truncated to a fixed dimension $M$ during training.

The input $\mathbf{z}$ is processed by a lightweight neural mapping function $f_\theta(\cdot)$, producing a softmax-normalized routing probability:

$$\boldsymbol{\pi} = \mathrm{softmax}(f_\theta(\mathbf{z})) \in \mathbb{R}^3$$

where $\boldsymbol{\pi} = (\pi_{\mathrm{cloud}}, \pi_{\mathrm{collab}}, \pi_{\mathrm{edge}})$. To encourage confident but flexible routing, we train the model with an entropy-regularized objective:

$$\mathcal{L}_{\mathrm{gating}} = \mathcal{L}_{\mathrm{task}} + \lambda \cdot \mathcal{H}(\boldsymbol{\pi}), \quad \mathcal{H}(\boldsymbol{\pi}) = -\sum_j \pi_j \log \pi_j$$

We set $\lambda = 0.4$ by default. The task loss $\mathcal{L}_{\mathrm{task}}$ is the cross-entropy between the predicted path and the ground-truth label from routing supervision.

At inference time, the most probable routing decision is selected deterministically via $\arg\max_j \pi_j$, ensuring stable behavior and avoiding sensitive prompt misrouting due to sampling randomness.

### A.5 Two-Layer Local Differential Privacy

We adopt a hierarchical entity-aware LDP mechanism:

- **Layer 1 (Type-level flipping)**: Randomly flips entity type with privacy budget $\epsilon_1$.
- **Layer 2 (Value-level masking)**: Replaces entity content using a type-specific replacement distribution with privacy budget $\epsilon_2$.

We fix total privacy budget $\epsilon = \epsilon_1 + \epsilon_2$, and allocate budget proportionally based on entity category risk (e.g., NAME vs. LOCATION). In our main experiments, we set $\epsilon = 1.0$, with $\epsilon_1 = 0.3$, $\epsilon_2 = 0.7$.

### A.6 Sketch Generation and Edge Denoising

The cloud-side sketch is generated via a few-shot prompting scheme. A set of semantic abstraction examples is prepended to the user query, guiding the LLM to produce concise, structured summaries. Sketches are limited to 100 tokens. On the edge, the sketch is passed to the SLM for final response generation via a template-based decoding scheme.

To guide the cloud-side LLM $\mathcal{G}_{\mathrm{cloud}}$ in generating sketches from perturbed prompts $P^*$, we design a diverse few-shot demonstration set $\mathcal{D}_{\mathrm{cloud}} = \{(P^{*(l)}, S^{(l)})\}_{l=1}^{k}$, covering multiple obfuscation scenarios including value replacement, category mismatches, and mixed-type distortions. Each sketch is structured, privacy-respecting, and intent-aligned, providing inductive signals for semantic abstraction.

Below are selected examples from three categories (Tourism, Medical, Banking), showing how entity-level noise is handled during sketch generation.
**Few-shot Demonstration Set for Sketch Generation**

**Tourism Examples**

- **Clean Input**: *"I plan to travel solo to Tokyo for two days; help me design my itinerary."* **Sketch**: Day 1: Arrival, local exploration; Day 2: Cultural visit, outdoor activity, local dining.

- **Category Obfuscated (Location → Person)**: *"I plan to travel solo to Emma for two days; help me design my itinerary."* **Sketch**: Day 1–2: Orientation, general city exploration, museum or park, casual dining, outdoor walk.

- **Mixed Obfuscation**: *"I plan to travel solo to JPMorgan for five days with my friend Elon."* **Sketch**: Day 1–5: Destination orientation, historical sites, social activities, structured itinerary with flexibility.

**Medical Examples**

- **Clean Input**: *"A 28-year-old female patient named Emma reports symptoms: headache, dizziness."* **Sketch**: Demographics noted; Symptoms logged; Plan: Neurological exam, imaging, specialist referral.

- **Category Obfuscated (Person → Organization)**: *"A 28-year-old female patient named IBM reports symptoms: headache, dizziness."* **Sketch**: Symptoms documented; Possible stress-related or neurological; Diagnostics recommended.

- **Value Obfuscated**: *"A 42-year-old female patient named Google reports symptoms: nausea and fatigue."* **Sketch**: Patient history recorded; Symptom profile analyzed; Diagnostics planned.

**Banking Examples**

- **Clean Input**: *"I want to file a dispute regarding a charge of $10 on my Chase card."* **Sketch**: Dispute initiated; Amount and institution recorded; Resolution timeline outlined.

- **Category Obfuscated (Bank → Person)**: *"I want to file a dispute regarding a charge of $10 on my Alice card."* **Sketch**: Dispute recorded; Possible merchant error; Review process initiated.

- **Mixed Obfuscation**: *"I want to file a dispute regarding a charge of $25 on my Emma account."* **Sketch**: Charge issue reported; Card/entity obfuscated; Process outlined for user assistance.

**Guidelines Enforced During Sketch Generation**

- Sketches are structured outlines, not full responses.
- Use concise, semantically aligned formats across categories.
- Obfuscation is realized via entity-value substitution, not placeholder tokens.
- Maintain abstraction consistency even under semantic distortion.
- Do not attempt to correct or infer true sensitive entities.
- Each sketch must remain interpretable and adequate for downstream denoising.

**Few-shot Demonstration Set for Edge-side Denoising and Response Generation**

Each example illustrates how the edge-side model reconstructs the final response $\hat{R}$ based on the original prompt $P$ and the cloud-generated sketch $S$ (derived from perturbed $P^*$). Sketches may omit sensitive entities, which are recovered using the clean prompt $P$ at the edge.

**Tourism Domain**

**Example 1 (Clean Sketch)**
**Prompt (P):** *"I plan to travel solo to Tokyo for two days; help me design my itinerary."*
**Sketch (S):** *"Day 1: Arrival and exploration; Day 2: Cultural visit and local experience."*
**Response ($\hat{R}$):** *"Sure! Here's a two-day itinerary for your Tokyo trip: Day 1: Morning – arrive and settle in; Afternoon – explore Shibuya and Meiji Shrine; Evening – try local ramen shops. Day 2: Morning – visit Asakusa and Senso-ji Temple; Afternoon – stroll in Ueno Park; Evening – Tokyo Tower view and sushi dinner."*

**Example 2 (Obfuscated Cloud Sketch)**
**Prompt (P):** *"I plan to travel solo to Tokyo for two days; help me design my itinerary."*
**Sketch (S):** *"Day 1: Arrival and destination orientation; Day 2: Cultural/historical visit, local dining."*
**Response ($\hat{R}$):** *"Absolutely! For your Tokyo trip: Day 1: Check into your hotel, explore Shinjuku, visit Tokyo Metropolitan Building. Day 2: Start with Meiji Shrine, then visit the Edo-Tokyo Museum, and enjoy dinner in a local izakaya."*

**Medical Domain**

**Example 3 (Clean Sketch)**
**Prompt (P):** *"A 28-year-old female patient named Emma reports symptoms: headache and blurred vision."*
**Sketch (S):** *"Demographics noted; Symptoms: Neurological; Plan: Imaging, specialist referral."*
**Response ($\hat{R}$):** *"Emma, a 28-year-old female, presents with headache and blurred vision. Recommend neurological exam, MRI brain imaging, and referral to neurology."*

**Example 4 (Obfuscated Cloud Sketch)**
**Prompt (P):** *"A 28-year-old female patient named Emma reports symptoms: headache and blurred vision."*
**Sketch (S):** *"Adult patient with neurological symptoms; plan includes imaging and specialist evaluation."*
**Response ($\hat{R}$):** *"Patient Emma, 28 years old, presents with headache and blurred vision. Schedule cranial MRI and consult neurology for further assessment."*

**Banking Domain**

**Example 5 (Clean Sketch)**
**Prompt (P):** *"I want to file a dispute for a $20 charge on my Chase card ending in 1234."*
**Sketch (S):** *"Dispute initiated for transaction; card identified; review and resolution process triggered."*
**Response ($\hat{R}$):** *"Understood. You've initiated a dispute for a $20 charge on your Chase card ending in 1234. We will review the transaction and notify you within 5–7 business days."*

**Example 6 (Obfuscated Cloud Sketch)**
**Prompt (P):** *"I want to file a dispute for a $20 charge on my Chase card ending in 1234."*
**Sketch (S):** *"Dispute filed for unknown card provider; steps include review and provisional credit."*
**Response ($\hat{R}$):** *"Thanks for reporting. You've filed a dispute on a $20 charge with Chase. We've started the review and will apply provisional credit if appropriate."*

**Edge-side Sketch-to-Response Prompting Guidelines**

- Use the original prompt $P$ to recover personalized semantics lost during LDP masking.

- Use sketch $S$ as structural guidance to organize the response flow.

- Ensure the reconstructed response preserves private context from $P$ while respecting the semantic plan from $S$.

- Do not override or contradict sketch structure unless required by critical details in $P$.

- Favor fluent, context-specific responses (not merely copying sketch items).

- Include named entities, temporal information, and personalized wording from $P$.

# B Theoretical Extensions and Proof Sketches

## B.1 Theorem: Privacy Composition for Multiple Sensitive Entities

**Definition.** Consider a user prompt containing $n$ sensitive entities $\{e_1, e_2, \ldots, e_n\}$ categorized into $m$ types. Let $C$ be the set of possible categories with $|C| = K_1$, and let each category $c \in C$ have a finite value domain $\mathcal{V}_c$ with maximum size $K_2$. The adaptive two-layer LDP mechanism $M$ is applied to each $e_i$ as:

1. Category-level randomized response $M^{(c)}$ with privacy budget $\epsilon_1$, producing $c_i^*$;
2. Value-level randomized response $M^{(v)}$ with budget $\epsilon_2$, producing $e_i^* \in \mathcal{V}_{c_i^*}$.

Define the combined mechanism $M^{(n)}$ as applying $M$ independently to each $e_i$, yielding $P^* = \{e_1^*, \ldots, e_n^*\}$.

**Theorem.** For any $\epsilon_1, \epsilon_2 \geq 0$, the mechanism $M^{(n)}$ satisfies $\epsilon_{\text{total}} = (\epsilon_1 + \epsilon_2)$-LDP with respect to any single-entity change, regardless of $n$. Furthermore, changing all $n$ entities leads to at most $n(\epsilon_1 + \epsilon_2)$-LDP.

**Proof Sketch.** For any adjacent prompts $P$ and $P'$ differing in only one entity $e_j$, we compute:

$$\frac{\Pr[M^{(n)}(P) = P^*]}{\Pr[M^{(n)}(P') = P^*]} = \frac{\Pr[M(e_j) = e_j^*]}{\Pr[M(e_j') = e_j^*]} \leq e^{\epsilon_1 + \epsilon_2}.$$

All other terms cancel due to independence. Hence $M^{(n)}$ is $(\epsilon_1 + \epsilon_2)$-LDP.

The group privacy extension follows directly: for $g$ modified entities, we apply the composition theorem, yielding privacy loss $\leq g(\epsilon_1 + \epsilon_2)$.

**Discussion.** This shows that the per-entity privacy guarantee holds even as $n$ grows, and total leakage scales linearly with the number of changed entities, as expected in parallel composition. Our scheme ensures bounded leakage per sensitive token without amplification.

## B.2 Theorem: Bounded Information Leakage in Entropy-Regularized Soft Gating

**Definition.** Let $\pi(x) = [\pi_{\text{cloud}}(x), \pi_{\text{collab}}(x), \pi_{\text{edge}}(x)]$ be the routing distribution computed via softmax over logits $f_\theta(x)$, optionally scaled by temperature $T$. The gating module is trained with:

$$L_{\text{gating}} = L_{\text{task}}(\pi(x)) + \lambda \cdot H(\pi(x)),$$

where $H(\pi) = -\sum_j \pi_j \ln \pi_j$ is the entropy regularizer, and $\lambda > 0$.

**Theorem.** The soft routing output $\pi(x)$ satisfies:

1. $\max_j \pi_j(x) \leq p_{\max} < 1$, i.e., no hard routing occurs;
2. Entropy $H(\pi(x)) \geq H_{\min} > 0$ for all $x$;
3. As $\lambda \to \infty$, $\pi(x) \to (\frac{1}{3}, \frac{1}{3}, \frac{1}{3})$, and information leakage $\to 0$;
4. As $\lambda \to 0$, $\pi(x)$ becomes peaked, and leakage approaches $\log 3$ bits.

**Proof Sketch.** At optimality, softmax ensures:

$$\pi_j(x) = \frac{\exp(-L_j/\lambda)}{\sum_k \exp(-L_k/\lambda)}.$$

The maximum component $\pi_{\max}(x)$ is bounded:

$$\pi_{\max}(x) \leq \frac{1}{1 + 2\exp(-\Delta/\lambda)} < 1,$$

where $\Delta = \max_{i \neq j}(L_j - L_i)$. The entropy is minimized when one $\pi_j$ is dominant:

$$H_{\min} = -[p_{\max} \ln p_{\max} + (1 - p_{\max}) \ln(1 - p_{\max})],$$

which remains $> 0$ under finite $\lambda$. Therefore, soft gating ensures bounded confidence and limits information leakage. In the limit $\lambda \to \infty$, all $\pi_j \to 1/3$ and $H(\pi) \to \ln 3$.

**Implication.** Entropy regularization introduces a privacy knob in gating decisions. By setting $\lambda$ appropriately, we bound how confidently the system routes based on sensitive inputs, mitigating privacy leakage through control flow itself.

| Model | PRISM | | | Edge-Only | | | Cloud-Only | | |
|---|---|---|---|---|---|---|---|---|---|
| | Ct.(s) | Ec.(J) | IQ. | Ct.(s) | Ec.(J) | IQ. | Ct.(s) | Ec.(J) | IQ. |
| *(L1)* GPT-4o + *(S1)* Phi-3.5-mini-3.5B | 8.29 | 683.83 | 7.00 | **15.98** | **1393.88** | 5.19 | 5.22 | **271.27** | **8.28** |
| *(L1)* GPT-4o + *(S2)* Qwen1.5-1.8B | **7.08** | **632.24** | 6.91 | 17.29 | 1540.33 | **5.59** | - | - | - |
| *(L1)* GPT-4o + *(S3)* Stablelm-2-zephyr-1.6B | 7.34 | 657.88 | 7.16 | 18.57 | 1627.46 | 4.94 | - | - | - |
| *(L1)* GPT-4o + *(S4)* Tinyllama-1.1B | 7.35 | 653.62 | 5.28 | 19.50 | 1734.24 | 4.62 | - | - | - |
| *(L2)* Qwen3-235B + *(S1)* Phi-3.5-mini-3.5B | 8.59 | 738.88 | **7.22** | - | - | - | **5.04** | 321.28 | 8.01 |
| *(L2)* Qwen3-235B + *(S2)* Qwen1.5-1.8B | 8.60 | 739.59 | 7.06 | - | - | - | - | - | - |
| *(L2)* Qwen3-235B + *(S3)* Stablelm-2-zephyr-1.6B | 8.00 | 693.13 | 7.19 | - | - | - | - | - | - |
| *(L2)* Qwen3-235B + *(S4)* Tinyllama-1.1B | 8.11 | 698.10 | 7.19 | - | - | - | - | - | - |

Table 2: Performance comparison of PRISM, Edge-Only, and Cloud-Only execution modes across various LLM–SLM pairs. For each metric (completion time, energy consumption, inference quality), we highlight the best result in **bold**.



\end{document}